\documentclass[final,twocolumn,pra,superscriptaddress,showpacs]{revtex4}
\usepackage{amssymb,amsmath,amsthm}

\newtheorem*{theorem*}{Theorem}
\newtheorem{lemma}{Lemma}

\providecommand{\abs}[1]{{\lvert{#1}\rvert}}
\providecommand{\bra}[1]{{\langle{#1}\rvert}}
\providecommand{\ket}[1]{{\lvert{#1}\rangle}}
\providecommand{\bracket}[2]{{\langle{#1}|{#2}\rangle}}

\DeclareMathOperator{\rank}{rank}
\DeclareMathOperator{\tr}{tr}
\DeclareMathOperator{\supp}{supp}

\newcommand{\ketbra}[2]{{\ket{#1}\!\bra{#2}}}
\newcommand{\pr}[1]{\ketbra{#1}{#1}}
\newcommand{\Hilbert}{{\mathcal H}}
\newcommand{\affDuess}{Institut f\"ur Theoretische Physik III,
             Heinrich-Heine-Universit\"at D\"usseldorf,
             D-40225 D\"usseldorf,
             Germany}

\begin{document}

\title{Commutator Relations Reveal Solvable Structures in Unambiguous State 
 Discrimination}
\author{M. Kleinmann}\email{kleinmann@thphy.uni-duesseldorf.de}
\author{H. Kampermann}
\affiliation{\affDuess}
\author{Ph. Raynal}
\affiliation{Quantum Information Theory Group,
             Institut f\"ur Theoretische Physik I, and
             Max-Plank-Forschungsgruppe,
             Institut f\"ur Optik, Information und Photonik,
             Universit\"at Erlangen-N\"urnberg,
             D-91058 Erlangen,
             Germany
            }
\author{D. Bru\ss}
\affiliation{\affDuess}

\begin{abstract}
We present a criterion, based on three commutator relations, that allows to 
 decide whether two self-adjoint matrices with non-overlapping support are 
 simultaneously unitarily similar to quasi-diagonal matrices, i.e.,  whether 
 they can be simultaneously brought into a diagonal structure with 
 $2\times2$-dimensional blocks.
Application of this criterion to unambiguous state discrimination provides a 
 systematic test whether the given problem is reducible to a solvable 
 structure.
As an example, we discuss unambiguous state comparison.
\end{abstract}
\maketitle

\section{Introduction}
The commutator of two self-adjoint operators, which act on a Hilbert space, is 
 a fundamental concept in quantum mechanics: two observables can be measured 
 without uncertainty if and only if their commutator vanishes.
This physical interpretation is connected to the mathematical fact that two 
 Hermitian matrices can be diagonalized simultaneously if and only if their 
 commutator is zero.
A natural question to ask is when two Hermitian matrices can be simultaneously 
 brought into a block-diagonal structure with blocks of the lowest non-trivial 
 size, namely size $2\times 2$.
Such structures are known as quasi-diagonal form and criteria for existence 
 have been studied in Ref.~\cite{Watters:1974LAA,Laffey:1977LAA}:
Watters \cite{Watters:1974LAA} showed that a family of normal matrices can be 
 simultaneously brought into a quasi-diagonal form if and only if each member 
 of the family commutes with the squared commutator of an element of the family 
 with any element from the algebra generated by the family.
(Thus, testing this criterion requires to show that infinitely many commutators 
 vanish.)
Laffey \cite{Laffey:1977LAA} studied a family with two members only.
He showed that when the matrices in the family are positive semi-definite, then 
 they are simultaneously unitarily similar to quasi-diagonal matrices if and 
 only if six certain commutators vanish.


The question of simultaneous quasi-diagonalizability has a physical application 
 in unambiguous discrimination of quantum states (see the next paragraph).
In that context, it is sufficient to deal with positive semi-definite operators 
 with non-overlapping supports (the support of an operator is the 
 orthocomplement of its kernel).
As we will show, this restriction leads to simpler commutator criteria.
In this paper we will give a constructive proof that, given two self-adjoint 
 operators with non-overlapping supports, they have a common block diagonal 
 structure of dimension two, if and only if a set of only three commutators 
 vanishes.
These commutators are also easier to calculate than the ones given in 
 Ref.~\cite{Laffey:1977LAA}, as the latter are of maximal order seven, while 
 the former are of maximal order five.

\emph{Unambiguous state discrimination} (USD) is a strategy for distinguishing 
 non-orthogonal quantum states without being allowed to make an error.
As it is impossible to discriminate non-orthogonal quantum states with unit 
 probability, the measurement has to have inconclusive outcomes.
The optimal USD strategy is the one that maximizes the success probability 
 (i.e., minimizes the probability to get an inconclusive result).
A different possibility to discriminate quantum states is called \emph{minimum 
 error discrimination}, where one minimizes the probability of making an error 
 in the state identification.

In this contribution we want to focus onto the first strategy, namely 
 unambiguous state discrimination.
For two density operators $\rho_1$ and $\rho_2$, acting on the Hilbert space 
 $\Hilbert$ of finite dimension, this task is described by a positive 
 operator-valued measure (POVM) on $\Hilbert$, consisting of three positive 
 operators $E_1$, $E_2$, and $E_?$, with $E_1+ E_2+ E_?= \openone$.
In order to make the discrimination unambiguous, the probability of wrong 
 identification must vanish, i.e. $\tr(E_1\rho_2)= 0$ and $\tr(E_2\rho_1)= 0$.
It is natural to allow $\rho_1$ and $\rho_2$ to have \emph{a priori} 
 probabilities $p_1$ and $p_2$, respectively, where $p_1> 0$, $p_2> 0$, and 
 $p_1+ p_2= 1$.
The open problem in USD is to find a POVM $\{E_1, E_2, E_?\}$ which maximizes 
 the success probability $p_\mathrm{succ}= p_1\tr(E_1\rho_1)+ 
 p_2\tr(E_2\rho_2)$.

While the optimal solution for minimum error discrimination of two mixed states 
 is already known for more than three decades \cite{Helstrom:1976}, the optimal 
 solution for unambiguous state discrimination has been found only for the pure 
 state case \cite{Jaeger:1995PLA} and certain special cases of mixed states 
 \cite{Bennett:1996PRA,Sun:2002PRA,Rudolph:2003PRA,Raynal:2003PRA,Herzog:2005PRA,Raynal:2005PRA,Bergou:2006PRA,Zhou:2007PRA,Raynal:2007XXX}.
A partial solution for unambiguous discrimination of mixed states is provided 
 via the reductions of the density operators by the space where perfect and/or 
 no USD is possible \cite{Raynal:2003PRA}.
Otherwise, known optimal USD measurements for mixed states mainly belong to the 
 class, where the problem can be decomposed into several pure state 
 discrimination tasks \cite{Bennett:1996PRA,Herzog:2005PRA,Bergou:2006PRA}.
A general representation of such states was recently discussed by Bergou 
 \emph{et al.} \cite{Bergou:2006PRA}.

It is not obvious, how to decide whether the given density operators possess 
 such a structure.
In this contribution we present a method that allows to systematically identify 
 if the optimal USD of two mixed states can be simplified to the pure state 
 task.

The paper is organized as follows.
In Sec.~\ref{s6636} we introduce the concept of common block-diagonal 
 structures of two operators.
We specifically consider the case of two-dimensional blocks, as the optimal 
 measurement in two dimensions is well-known.
Simple commutator relations are presented to check for the existence of such a 
 structure.
In Sec.~\ref{s12765} we discuss whether the block structures are preserved by 
 the reductions.
Finally, we study the example of unambiguous state comparison 
 \cite{Barnett:2003PLA,Rudolph:2003PRA,Chefles:2004JPA,Herzog:2005PRA,Kleinmann:2005PRA}, 
 to illustrate the power of the commutator test.

\section{Block-diagonal structures}\label{s6636}
\subsection{Independent orthogonal subspaces in USD}
In Ref.~\cite{Bennett:1996PRA} Bennett \emph{et al.} analyzed the parity check 
 for a string of qubits, i.e., the question whether a sequence composed of 
 states that are either $\ket{\psi_0}$ or $\ket{\psi_1}$, with $0< 
 \abs{\bracket{\psi_0}{\psi_1}}< 1$, contains an even or odd number of 
 occurrences of $\ket{\psi_1}$.
This task is equivalent to the unambiguous discrimination of two certain mixed 
 states.
After a suitable (symmetric) choice of a basis these mixed states turned out to 
 share the same block-diagonal shape, with each block $\blacksquare$ 
 symbolizing a $2\times2$ matrix:
\begin{equation}
 \rho_1= \begin{pmatrix}
  \blacksquare&            &      \\
              &\blacksquare&      \\
              &            &\ddots
 \end{pmatrix},\quad
 \rho_2= \begin{pmatrix}
  \blacksquare&            &      \\
              &\blacksquare&      \\
              &            &\ddots
  \end{pmatrix}.
\end{equation}
The authors of Ref.~\cite{Bennett:1996PRA} argued that due to this structure an 
 optimal solution to the discrimination problem can be obtained by the simple 
 composition of the optimal solutions in each block.
The optimal solution in two dimensions is known, since only in the case of two 
 pure states the solution is not obvious and this case was solved by Jaeger and 
 Shimony \cite{Jaeger:1995PLA}.

Our aim is to provide a systematic method for finding such structures.
We start with a formal definition of a block-diagonal structure:
For a set of operators $\mathcal O$, a \emph{common block-diagonal structure}
 (CBS) is a projection-valued measure $\{\Pi_k\}$ such that all operators in 
 $\mathcal O$ commute with any $\Pi_k$.
In other words, if the operators in $\mathcal O$ have a CBS, they can be 
 simultaneously decomposed in orthogonal subspaces, and a von-Neumann 
 measurement $\{\Pi_k\}$ projects onto these subspaces.
Having the measurement outcome ``$k$'', the support of the states is reduced to 
 $\Pi_k\Hilbert$ (the image of $\Pi_k$).
Thus one can focus on performing the optimal measurement in this subspace.

A common block-diagonal structure is \emph{at most $n$-dimensional} if the rank 
 of all $\Pi_k$ is at most $n$.
In particular the existence of an at most one-dimensional CBS for a set 
 $\mathcal O$ of normal operators (a normal operator is an operator that 
 commutes with its adjoint) is equivalent to the existence of a common basis, 
 in which all operators in $\mathcal O$ are diagonal.
It is well-known (cf. e.g. Chapter~IX, Theorem~11 in 
 Ref.~\cite{Gantmacher:1959}) that for normal operators this is possible if and 
 only if all operators in $\mathcal O$ mutually commute.
We will present a commutator criterion to verify whether two operators have an 
 at most \emph{two}-dimensional CBS (2d-CBS).
This criterion, which is simpler (from an operational point of view) then the 
 one introduced by Laffey \cite{Laffey:1977LAA}, is valid in the case of 
 non-overlapping support only, but is sufficiently general in order to detect 
 any two-dimensional block structure in the case of USD.

\subsection{Diagonalizing Jordan bases: Definition and existence}
Let us first relate the idea of a 2d-CBS to a concept that is widely used in 
 the analysis of USD, namely the concept of \emph{Jordan (or canonical) bases} 
 of subspaces (cf. e.g. Ref.~\cite{Stewart:1990}):
Let $P_A$ and $P_B$ be self-adjoint projectors.
Then by virtue of the singular value decomposition, one can find orthonormal 
 bases $\{\ket{\alpha_i}\}$ of $P_A\Hilbert$ and $\{\ket{\beta_j}\}$ of 
 $P_B\Hilbert$, such that
\begin{subequations}\label{jordan12}\begin{gather}
 \bracket{\alpha_i}{\beta_j}\equiv \bra{\alpha_i}P_A P_B \ket{\beta_j}= 0
 \quad\text{for $i \ne j$,}\\
\intertext{while for $i\le \min\{\rank P_A, \rank P_B\}$,}
 \bracket{\alpha_i}{\beta_i}\equiv \bra{\alpha_i}P_A P_B \ket{\beta_i}\equiv
    \cos\vartheta_i\ge 0
\end{gather}\end{subequations}
 for some $0\le \vartheta_i\le \pi/2$.
The bases $\{\ket{\alpha_i}\}$ and $\{\ket{\beta_j}\}$ are called \emph{Jordan 
 bases} of the subspaces $P_A\Hilbert$ and $P_B\Hilbert$ and $\{\vartheta_i\}$ 
 are the corresponding (unique) \emph{Jordan angles}.
The first equation expresses the bi-orthogonality of the Jordan bases.
Note that in the case of degenerate Jordan angles (i.e., not all Jordan angles 
 are different) or if $\abs{\rank P_A- \rank P_B}\ge 2$, the Jordan bases are 
 not unique.

For the analysis of USD, it turns out to be fruitful to consider density 
 operators, which are diagonalized by a pair of Jordan bases 
 \cite{Bergou:2006PRA}:
For two normal operators $A$ and $B$, \emph{diagonalizing Jordan bases} are 
 Jordan bases of $\supp A$ and $\supp B$, which diagonalize $A$ and $B$, 
 respectively.
Of course, such diagonalizing Jordan bases do not always exist.
As mentioned in Ref.~\cite{Herzog:2007PRA}, the existence of such bases 
 involves the presence of a 2d-CBS.
However the converse is in general not true.
It is possible that already in two dimensions no pair of diagonalizing Jordan 
 bases exists:
Consider the positive semi-definite matrices
\begin{equation}
 A=\begin{pmatrix} 1 & 0 \\ 0 & 2 \end{pmatrix} \quad\text{and}\quad
 B=\begin{pmatrix} 1 & 1 \\ 1 & 1 \end{pmatrix}.
\end{equation}
Then up to some complex phases, the only orthonormal basis of $\supp A$ that 
 diagonalizes $A$ is the canonical basis $\{(1,0),(0,1)\}$ while $\supp B$ is 
 spanned by $(1,1)$.
But $(1,1)$ is orthogonal to neither $(1,0)$ nor $(0,1)$, i.e., no 
 diagonalizing Jordan bases exist.

The exact relation between 2d-CBS and diagonalizing Jordan bases is given by 
 the following
\begin{lemma}\label{p16663}
Let $A$ and $B$ be normal operators acting on $\Hilbert$.
Then diagonalizing Jordan bases of $A$ and $B$ can be found if and only if a 
 2d-CBS of $A$ and $B$ exists and $[A,ABA]= 0$ and $[B,BAB]= 0$.
\end{lemma}
\begin{proof}
Assume that diagonalizing Jordan bases of $A$ and $B$ exist.
Then their structure readily provides an appropriate 2d-CBS.
Furthermore, by writing $A$ and $B$ in diagonalizing Jordan bases, i.e., 
 $A=\sum_i a_i\pr{\alpha_i}$ and $B=\sum_j b_j\pr{\beta_j}$, and using 
 Eqs.~\eqref{jordan12}, it is easy to verify that $[A,ABA]=0$ and $[B,BAB]= 0$ 
 holds.

For the contrary it is enough to prove the assertion in each subspace 
 $\Pi_k\Hilbert$, where $\{\Pi_k\}$ is a 2d-CBS of $A$ and $B$.
Since $A$ and $B$ commute with all projectors $\Pi_k$, in each subspace the 
 operators $A_k\equiv \Pi_kA\Pi_k$ and $B_k\equiv \Pi_kB \Pi_k$ are again 
 normal.
First suppose that $A_k$ has maximal rank, i.e., rank two.
Since $A_k$ has full rank in $\Pi_k\Hilbert$, the condition $0= \Pi_k [A, ABA] 
 \Pi_k= A_k[A_k, B_k]A_k$ is equivalent to $\Pi_k[A_k, B_k]\Pi_k\equiv [A_k, 
 B_k]= 0$, i.e., both operators can be diagonalized simultaneously and hence in 
 particular diagonalizing Jordan bases exist.
(An analogous argument holds if $B_k$ has maximal rank.)
The remaining non-trivial case is that both operators have rank one, in which 
 case the diagonalizing Jordan bases are given by the vector spanning the 
 support of each operator.
\end{proof}

\subsection{Construction of diagonalizing Jordan bases}\label{secconsjordan}
It is a simple observation, that if diagonalizing Jordan bases for two normal 
 operator $A$ and $B$ exist, then necessarily all commutators of the structure 
 $[A, ABA]$, $[A, AB^2A]$ and so forth vanish (see proof of Lemma 1).
In the following Lemma we will state that certain of these commutators already 
 suffice to explicitly construct a pair of diagonalizing Jordan bases.
\begin{lemma}
Let $A$ and $B$ be self-adjoint operators on $\Hilbert$ with $[A, ABA]= 0$, 
 $[A, AB^2A]= 0$ and $[B, BA^2B]= 0$.
Furthermore denote by $\{\ket k\}$ an orthogonal basis of $\supp A$ which 
 simultaneously diagonalizes $A$, $ABA$ and $AB^2A$.

Then there exists vectors $\{\ket\nu\}$, such that (up to normalization), 
 $\{A\ket k\}$ and $\{BA\ket k\}\cup \{\ket\nu\}$ are diagonalizing Jordan 
 bases of $A$ and $B$.
\end{lemma}
\begin{proof}
First note that all vectors $BA\ket k$ are mutually orthogonal (or trivial), 
 since the basis $\{\ket k\}$ diagonalizes $AB\, BA$.
Now consider the following expression
\begin{equation}\label{e19032}\begin{split}
 w_k B(BA\ket k)&= BBA(A B A\ket k)\\
    &= BA^2 B\, B A\ket k\\
    &= v_k BA\ket k,
\end{split}\end{equation}
 where $w_k$ denotes the eigenvalue of $ABA$ for $\ket k$ and $v_k$ denotes the 
 eigenvalue of $AB^2A$ for $\ket k$.
In the second step we used $[B,BA^2B]= 0$.
The right hand side can only vanish, if $BA\ket k= 0$.
Hence due to Eq.~\eqref{e19032}, $BA\ket k\in \supp B$ is either trivial or is 
 an eigenvector of $B$.
Furthermore one readily finds eigenvectors $\ket\nu\in \supp B$ of $B$ that 
 complete the orthogonal basis of $\supp B$.
These vectors are also orthogonal to all $A\ket k$, since by construction 
 $b_\nu \bra \nu A \ket k= \bra \nu B A \ket k= 0$, where $b_\nu\ne 0$ is the 
 eigenvalue of $B$ for $\ket \nu$.
It remains to verify, that $\{A\ket k\}$ and $\{BA\ket k\}$ are bi-orthogonal.
But this follows from the fact that $\{\ket k\}$ diagonalizes $A B A$.
\end{proof}

Note that it is straight-forward to extend this Lemma to normal operators.
However, we are mainly interested in application for USD and hence specialize 
 the results of this section in the following form:
\begin{theorem*}
For two self-adjoint $A$ and $B$ operators on a Hilbert space of finite 
 dimension with $\supp A\cap \supp B= \{0\}$ the following statements are 
 equivalent:
(i) $A$ and $B$ have a 2d-CBS.
(ii) Diagonalizing Jordan bases of $A$ and $B$ exist.
(iii) $[A,ABA]= 0$, $[B,BA^2B]= 0$ and $[A,AB^2A]= 0$.
\end{theorem*}
\begin{proof}
Remember that (ii) $\Rightarrow$ (i) follows from the structure of Jordan bases 
 (see Lemma 1), and also (ii) $\Rightarrow$ (iii) is a consequence of the 
 properties of Jordan bases (see Sec.~\ref{secconsjordan}).
The implication (iii) $\Rightarrow$ (ii) was proven in Lemma 2.
It remains to show that from (i) follows (ii).
Due to Lemma~\ref{p16663} this reduces to showing that $[A, ABA]= 0$ and $[B, 
 BAB]= 0$ for the case where (i) holds and $\supp A\cap \supp B= \{0\}$.
The condition of non-overlapping supports implies together with (i), that 
 $\rank(A_k)+ \rank(B_k) \le 2$, where $A_k = \Pi_k A \Pi_k$ and $ B_k =\Pi_k B 
 \Pi_k$, and $\{\Pi_k\}$ is a 2d-CBS of $A$ and $B$.
If either $\rank(A_k)$ or $ \rank(B_k)$ is zero, the commutators $[A_k, 
 A_kB_kA_k]$ and $[B_k, B_kA_kB_k]$ vanish trivially.
They are also equal to zero for the remaining case of $\rank(A_k)= 1= \rank 
 (B_k)$.
\end{proof}

As soon as the supports of $A$ and $B$ overlap, in general, none of the 
 commutators in the above Theorem vanishes.
But in such a situation one can make use of the fact that in two dimensions, 
 the square of all commutators of the form $[A_k,B_k]$, $[A_k,B_k^2]$ and so 
 forth is proportional to the identity operator.
Laffey \cite{Laffey:1977LAA} showed that for positive operators the following 
 set of commutators, given below, are already sufficient to prove the existence 
 of a 2d-CBS:

\emph{Two positive semi-definite operators $A$ and $B$ have a 2d-CBS if and 
only if \cite{Laffey:1977LAA}
\begin{equation}
\begin{split}
[[A,  B  ]^2,A]= 0,&\quad [[B,  A  ]^2,B]= 0,\\
[[A,  B^2]^2,A]= 0,&\quad [[B,  A^2]^2,B]= 0,\\
[[A^2,B  ]^2,A]= 0,&\quad [[B^2,A  ]^2,B]= 0.
\end{split}
\end{equation}
}

\section{Application to USD}\label{s12765}
We now want to apply the above analysis to unambiguous discrimination of two 
 mixed states $\rho_1$ and $\rho_2$.
We denote the combination of the density operator and the according \emph{a 
 priori} probability by $\gamma_\mu= p_\mu \rho_\mu$, such that $\tr\gamma_\mu 
 <1$ ($\mu= 1,2$).
For technical reasons (see the map $\tau_0$ below) we also allow that the 
 \emph{a priori} probabilities do not sum up to one, $\tr(\gamma_1)+ 
 \tr(\gamma_2) \le 1$.

\subsection{Preservation of block structures under reduction of USD}
In the above Theorem the density operators need to satisfy the condition 
 $\supp\gamma_1\cap \supp\gamma_2=\{0\}$, which in general is not the case.
The first reduction theorem in Ref.~\cite{Raynal:2003PRA}, however, shows how 
 to reduce any USD problem to that specific form.
But one could imagine, that this reduction might destroy an already present 
 2d-CBS, so that the combination of the first reduction theorem together with 
 the above Theorem would fail to detect certain block-diagonal structures.
As we will see here, this is not the case and the application of any of the 
 reductions in Ref.~\cite{Raynal:2003PRA} preserves any CBS.

We repeat the reductions of \cite{Raynal:2003PRA} in the language of 
 projectors:
For a pair of positive operators $(\gamma_1, \gamma_2)$, let $\tau_0$ be the 
 (non-linear) mapping
\begin{equation}
 \tau_0\colon (\gamma_1, \gamma_2)\mapsto (\gamma_1^0, \gamma_2^0),
\end{equation}
 where $\gamma_\mu^0$ (with $\mu= 1,2$) is the projection of $\gamma_\mu$ onto 
 $(\ker\gamma_1+ \ker \gamma_2)$.
In a similar fashion we define $\tau_\nu\colon (\gamma_1, \gamma_2)\mapsto 
 (\gamma_1^\nu, \gamma_2^\nu)$ (with $\nu = 1,2$) where
\begin{equation}
 \gamma_\mu^\nu= P_\nu \gamma_\mu P_\nu+
              (\openone- P_\nu) \gamma_\mu (\openone- P_\nu).
\end{equation}
Here, $P_1$ is the self-adjoint projector onto $(\ker\gamma_1+ \supp\gamma_2)$ 
 and $P_2$ the projection onto $(\ker\gamma_2+ \supp\gamma_1)$.
The reduction theorems in Ref.~\cite{Raynal:2003PRA} now read as follows:

\emph{For $\tau\in \{\tau_0, \tau_1, \tau_2\}$, the pair $(\gamma_1,\gamma_2)$ 
 and the reduced pair $\tau(\gamma_1, \gamma_2)$ can be unambiguously 
 discriminated with the same success probability \cite{Raynal:2003PRA}.}

What is relevant for our considerations is the fact that no reduction can 
 destroy any CBS, i.e., a CBS $\{\Pi_k\}$ of $(\gamma_1, \gamma_2)$ is also a 
 CBS of $\tau(\gamma_1,\gamma_2)$ for all $\tau\in \{\tau_0, \tau_1, \tau_2\}$.
In order to see this, it is enough to show that any of the projectors $P_0$, 
 $P_1$, and $P_2$  (with $P_0$ denoting the projector onto 
 $\ker\gamma_1+\ker\gamma_2$) commutes with all $\Pi_k$.
But this follows from the fact, that the range of each of the projectors is the 
 support of an operator, that commutes with all $\Pi_k$ (namely, $P_0\Hilbert= 
 \supp(2\openone- G_1- G_2)$, $P_1\Hilbert= \supp(\openone- G_1+ G_2)$ and 
 $P_2\Hilbert= \supp(\openone- G_2+ G_1)$, where $G_\mu$ is the projector onto 
 $\supp\gamma_\mu$).
Note however, in contrast, that a CBS of $\tau(\gamma_1, \gamma_2)$ is not 
 necessarily a CBS of $(\gamma_1, \gamma_2)$, thus a reduction may give rise to 
 new block-diagonal structures.

In order to check for a 2d-CBS it is necessary to first apply the reduction 
 $\tau_0$.
If the reductions $\tau_1$ and $\tau_2$ are -- from an operational point of 
 view -- feasible, then it also worth to apply those, since new 2d-CBS may 
 arise.

\subsection{Example: State comparison}
We consider a special case of unambiguous state comparison ``two out of $N$'' 
 as defined in Ref.~\cite{Kleinmann:2005PRA}.
A source emits pure states $\{\ket{\psi_1}, \cdots, {\ket{\psi_N}}\}$, each of 
 which appears with equal a priori probability $\frac1N$.
We further assume that all states have the same (real) mutual overlap, 
 $\bracket{\psi_i}{\psi_j}= \cos\vartheta$ for $i\ne j$.
Given two of these pure states, the aim is to decide unambiguously whether the 
 states are identical or not.
This task is equivalent to the discrimination of
\begin{eqnarray}
 \gamma_1&=& \frac1{N^2} \sum_{k=1}^N \pr{\psi_k\psi_k}, \\
 \gamma_2&=& \frac1{N^2} \sum_{k\ne l}^N \pr{\psi_k\psi_l}.
\end{eqnarray}

From the definition it follows that $\supp\gamma_1\cap \supp\gamma_2= \{0\}$.
Thus we can directly apply the Theorem of Sec.~\ref{secconsjordan}, i.e., we 
 test whether it is true that $[\gamma_1, \gamma_1\gamma_2\gamma_1]= 0$, 
 $[\gamma_1, \gamma_1\gamma_2^2\gamma_1]= 0$ and $[\gamma_2, \gamma_2 
 \gamma_1^2\gamma_2]= 0$.
For the first two commutators, it is sufficient to verify that 
 $\omega_{kl}\equiv \bra{\psi_k\psi_k} [\cdots]\ket{\psi_l\psi_l}= 0$ for any 
 $k$ and $l$.
Here, $[\cdots]$ stands for any of the first two commutators.
Obviously we have $\omega_{kl}= -(\omega_{lk})^*$ for all $k$ and $l$, and 
 since all overlaps are real, $\omega_{kk}=0$.
Due to the high symmetry, all $\omega_{kl}$ with $k\ne l$ must be equal.
In particular $\omega_{kl}= \omega_{lk}= -(\omega_{kl})^*$, and again due to 
 reality of the overlaps, $\omega_{kl}= 0$ must hold.

It remains to test, whether $[\gamma_2, \gamma_2\gamma_1^2\gamma_2]= 0$.
This is equivalent to showing that $\gamma_2[\gamma_2+\gamma_1, 
 \gamma_1^2]\gamma_2 =0$ or to showing that
\begin{equation}
 \gamma_2(\gamma_2+\gamma_1)\gamma_1^2\gamma_2=
  \sum_{i,j;p,q} \ketbra{\psi_i\psi_j}{\psi_p\psi_q} A_{ij,pq}
\end{equation}
 is self-adjoint.
For $i\ne j$ and also $p\ne q$, we have
\begin{equation}
 A_{ij,pq}= \sum_{k,l,n,m} c_{ik} c_{jl} c_{kn} c_{ln} c_{nm}^2 c_{mp} c_{mq},
\end{equation}
 with $c_{ij}\equiv \bracket{\psi_i}{\psi_j}= \cos\vartheta+ 
 (1-\cos\vartheta)\delta_{ij}$.
Otherwise, $A_{ij,pq}= 0$.
First we find
\begin{align}
 \sum_k c_{ik} c_{kn}&\propto \delta_{in}+ \mu,\\
\intertext{with some constant $\mu$. Also, for $p\ne q$,}
 \sum_m c_{nm}^2 c_{mp} c_{mq}&\propto \delta_{nq}+ \delta_{np}+ \sigma,
\end{align}
 where $\sigma$ is another constant.
Hence for $i\ne j$ and $p\ne q$ we have
\begin{equation}\begin{split}
 A_{ij,pq}&\propto \sum_n (\delta_{in}+\mu)(\delta_{jn}+ \mu)
    (\delta_{np}+ \delta_{nq}+ \sigma) \\
    &\propto \delta_{ip}+ \delta_{iq}+ \delta_{jp}+ \delta_{jq}+
        \mathit{const.}
\end{split}\end{equation}
In particular $A_{ij,pq}= A_{pq,ij}\equiv (A_{pq,ij})^*$ holds, which 
 demonstrates that $\gamma_2(\gamma_2+ \gamma_1)\gamma_1^2\gamma_2$ is 
 self-adjoint and therefore $\gamma_2[\gamma_2+ \gamma_1,\gamma_1^2]\gamma_2 
 =0$.

Thus we have shown that the symmetric state comparison ``two out of $N$'' can 
 be reduced to pure state discrimination.
Note that this statement is in general not true for state comparison ``$C$ out 
 of $N$'', with $C>2$, i.e., the question whether $C$ states taken from a set 
 of $N$ states (with equal overlaps) are identical or not.
In this case the third commutator does not vanish before the reductions, and 
 the corresponding state discrimination problem is not necessarily simplified 
 to the pure state case.

\section{Conclusions}
In many practical situations of unambiguous state discrimination (USD) the pair 
 of states that one wants to discriminate has a high symmetry which naturally 
 gives rise to a two-dimensional common block-diagonal structure (2d-CBS) 
 \cite{Bennett:1996PRA,Herzog:2005PRA,Bergou:2006PRA}.
In this situation the optimal USD measurement has the very same 2d-CBS 
 \cite{Bergou:2006PRA}, where each block basically is given by the pure state 
 solution of Jaeger and Shimony \cite{Jaeger:1995PLA}.

Here, we provided a tool to systematically identify whether a given USD task 
 possesses such a structure.
With the commutator relations presented in this paper it is easy to test 
 whether a 2d-CBS for two self-adjoint operators with non-overlapping support 
 exists.
In order to derive these commutator relations, we studied the connection 
 between the existence of a 2d-CBS and of diagonalizing Jordan bases.
This also led to an explicit construction procedure for such bases.

We showed that the reduction method \cite{Raynal:2003PRA} for USD can only 
 generate, but not destroy a 2d-CBS.
Thus, applying the reductions as a first step ensures that the condition of 
 non-overlapping support of the two operators is fulfilled.

We demonstrated the strength of the simple commutator relations by considering 
 unambiguous state comparison 
 \cite{Barnett:2003PLA,Rudolph:2003PRA,Chefles:2004JPA,Herzog:2005PRA,Kleinmann:2005PRA}, 
 where it is easy to show that in completely symmetric situations for the 
 specific case ``two out of $N$'' a 2d-CBS exists.

Outlook:
Note that the commutator relations in the Theorem of Sec.~\ref{secconsjordan} 
 are not symmetric in both operators (i.e., the missing commutator $[BAB, B]$ 
 already vanishes).
It would be interesting to understand the reason for  this asymmetry.
Furthermore, it would be useful to extend this concept to be applicable to more 
 than two operators and also to the detection of larger block-diagonal 
 structures (with respect to USD, e.g. four-dimensional structures would be 
 interesting).
In order to be operational, this would mean to extend the work by Watters 
 \cite{Watters:1974LAA} and Shapiro \cite{Shapiro:1979LAA} (generalization to 
 blocks of arbitrary dimension) and finding  a \emph{finite} set of commutators 
 with possibly low order.

\begin{acknowledgments}
The authors would like to thank A. Doherty, N. L\"utkenhaus, T. Meyer, D.  
 Petz, and R. Unanyan for valuable discussions, and S. Wehner for pointing out 
 references \cite{Watters:1974LAA} and \cite{Shapiro:1979LAA}.
This work was partially supported by the EU Integrated Projects SECOQC and 
 SCALA.
\end{acknowledgments}

\bibliography{the}
\end{document}